\documentclass{article}

\usepackage{amsmath,enumerate,fancyhdr,amssymb,amsfonts,framed}
\usepackage{graphicx}           
\usepackage{xifthen}
\usepackage{myDefs}
\usepackage[numbers]{natbib}

\newcommand{\tp}[2]{\left((#1),#2\right)}

\newcommand{\comment}[1]{}

\newcommand{\abs}[1]{\vert #1 \vert}

\newcommand{\sqk}{\sqrt{\ln(k)}}
\newcommand{\lln}{\ln(\ln(k))}

\renewcommand{\H}{\mathcal{H}}
\newcommand{\E}{\mathbf{E}}
\renewcommand{\D}{\mathcal{D}}
\newcommand{\Prob}{\mathbf{P}}
\newcommand{\cH}{\lceil \mathcal{H}_k \rceil}
\newcommand{\ncH}{ \mathcal{H}_k }
\newcommand{\R}{{\bf R}}
\newcommand{\Var}{\mathbf{Var}}
\newcommand{\core}{\mathrm{Core}}
\newcommand{\cb}{\mathrm{CoreRev}}
\newcommand{\vcgr}{\mathrm{VcgRev}}
\newcommand{\mv}{\mathrm{MV}}

\title{Core-competitive Auctions}
\date{}
\author{GAGAN GOEL\thanks{Google Research}\quad 
	MOHAMMAD REZA KHANI\thanks{University of Maryland}\quad
	RENATO PAES LEME\thanks{Google Research}
}
\begin{document}

\maketitle

\begin{abstract}

One of the major drawbacks of the celebrated VCG auction is its low (or zero) revenue even when the agents have
high value for the goods and a {\em competitive} outcome would have generated a
significant revenue. A competitive outcome is one for which it is impossible for
the seller and a subset of buyers to `block' the auction by defecting and
negotiating an outcome with higher payoffs for themselves. This corresponds to
the well-known concept of {\em core} in cooperative game theory.

In particular, VCG revenue is known to be not competitive when the goods being
sold have {\em complementarities}. Complementary goods are present in many
application domains including spectrum, procurement, and ad auctions. The
absence of good revenue from VCG auction poses a real hurdle when trying to
design auctions for these settings.  Given the importance of these application
domains, researchers have looked for alternate auction designs. One important research direction that has come from this line of thinking is that of the
design of {\em core-selecting auctions} (See  Ausubel and Milgrom, Day
and Milgrom, Day and Cramton, Ausubel and Baranov).
Core-selecting auctions are combinatorial auctions whose outcome
implements competitive prices even when the goods are complements. While these
auction designs have been implemented in practice in various scenarios and are
known for having good revenue properties, they lack the desired
incentive-compatibility property of the VCG auction. A bottleneck here is an
impossibility result showing that there is no auction that simultaneously
achieves competitive prices (a core outcome) and incentive-compatibility.

In this paper we try to overcome the above impossibility result by asking the
following natural question: is it possible to design an incentive-compatible
auction whose revenue is comparable (even if less) to a competitive
outcome? Towards this, we define a notion of {\em core-competitive} auctions. We
say that an incentive-compatible auction is $\alpha$-core-competitive if its
revenue is at least $1/\alpha$ fraction of the minimum revenue of a
core-outcome. We study one of the most commonly occurring setting in
Internet advertisement with complementary goods, namely that of the
Text-and-Image setting. In this setting, there is an ad slot which can be filled
with either a single image ad or $k$ text ads. We design an $O(\ln \ln k)$
core-competitive randomized auction and an $O(\sqk)$ competitive deterministic
auction for the Text-and-Image setting. We also show that both factors are
tight.

\end{abstract}


\section{Introduction}

\comment{
\begin{quotation}
Despite the enthusiasm that the Vickrey mechanism and its extensions generate
among economists, practical applications of Vickrey’s design are rare at best. 
\end{quotation}

Above are the quotes from a paper titled {\em “The lovely but the lonely Vickrey auction” } written by two notable economists Ausubel and Milgrom \cite{AM06}. This forms the starting point of our work as well as we seek to design auctions for Internet advertisement where we encounter bidders with complementarities -- a setting where Vicrey and its generalizations (VCG) are known to be ill-suited from a practical auction design perspective. \\
}

The VCG mechanism is a powerful mechanism that achieves an efficient
outcome in an incentive compatible manner for a variety of scenarios. The simplicity of the VCG mechanism raised our hopes of wide
application of this elegant theory in practice. However, it has been noted
in the recent past that the applicability of VCG auction beyond the simple 
case of multiple homogeneous goods has remained limited.
Ausubel and Milgrom \cite{AM02} offer an explanation of why VCG in its
purest form is often unsuitable to be used in practice. They write:

\begin{quotation}
[...] \comment{In auctions of public
assets,} higher revenues also improve efficiency, since auction revenues can displace
distortionary tax revenues. [...]
Probably the most important disadvantage of the Vickrey auction is that the revenues
it yields can be very low or zero, even when the items being sold are quite valuable.
\end{quotation}

\comment{One primary reason for the limited use of the general VCG auction in practice is
that VCG gives low (or zero) revenue even when the items being sold are quite
valuable and there is sufficient competition in the market.}

To illustrate this point of low or zero revenue, consider the following example
from spectrum auctions (taken from \cite{AM02,ausubel2006lovely}): consider 3
bidders who are participating in an auction for two spectrum licenses: the first
bidder is willing to pay 2 billion for the package of 2 licenses while each of
the other two bidders is willing to pay 2 billion for any individual license. The
VCG outcome allocates to the second and third bidder, and charges a payment of zero to each of them. This is because the externality each winning bidder imposes on the rest of the bidders is
zero. Note that, one can hardly blame the lack of revenue to the absence of competition; if one were
to treat it as a market equilibrium problem and compute market clearing prices
(say by means of a tatonnement procedure), the revenue would be non-trivial.

Thus, one natural question to ask is, for an auction outcome, how to formally say
that it achieves {\em a competitive revenue}? To answer this,
\cite{AM02} introduced the notion of a {\em core} outcome in an auction setting.
The notion of core is a fundamental and well-known notion in cooperative game
theory and represents a way to share the utility produced by a group of players
in a manner that no sub-group of players would want to deviate. In an auction
setting, a set of winning buyers and their payments are said to be a core
outcome if no sub-group of losing bidders can propose to the auctioneer (seller)
an alternative higher-revenue outcome. For example, in the license example, the
outcome implemented by VCG is not in the core since the first bidder (who wanted
to purchase two licenses) could negotiate with the auctioneer that the licenses
should be allocated to him for any price larger than zero.  On the other hand,
the outcome which allocates one license each to players 2 and 3 and charges each of
them 1 billion is in the core, since in this case there is no alternative outcome that the
first player can propose to the auctioneer which would be beneficial for both.

It is noteworthy that when the goods are {\em substitute}, the VCG outcome is a
core outcome, and VCG revenue equals the core-outcome with the minimum revenue
(the set of core outcomes is not unique) \cite{AM02}. However, if the goods are not
substitutes, the VCG outcome may lie outside the core.
In fact, as shown in the above example, VCG revenue can be arbitrarily lower when
compared to the minimum-revenue core outcome.

So can one design incentive-compatible auctions whose outcome is always in the
core? Unfortunately, one can show that it is impossible to design an auction
that (a) achieves a core outcome, and (b) has truth-telling as a dominant
strategy equilibrium. So we must either relax (a) or (b). In Ausubel and Milgrom
\cite{AM02}, the authors relax (b), and give a family of ascending package
auctions (called {\em core-selecting} auctions) which are not truthful but whose
equilibrium outcome is a core outcome. These auctions have been extremely
successful in practice -- variations of these were used in spectrum-license
auctions in the United Kingdom,  Netherlands, Denmark, Portugal, and Austria,
and in the auction of landing-slot rights in the three New York City airports.
See \cite{day2012quadratic} for a complete discussion.

The focus of this paper is on applications in Internet ad auctions (we will call
them ad auctions from now on). There are several ad auction scenarios which are
modeled as goods with complementarities. As a case study for our work, we use a
very common scenario in ad auctions which has complementarities, namely that of
Text-and-Image ad auction. In a Text-and-Image ad auction scenario an ad slot on
a page can either accommodate $k$ text ads (which are the traditional ads
displayed next to search results) or one large image-ad. Notice that the example
by Ausubel and Milgrom can be reproduced exactly in this setting by setting
$k=2$.

What auction should we use for the Text-and-Image setting? The core-selecting
auction of \cite{AM02} is not a good choice for this setting as the ascending
package auctions are interactive procedures in which bidders submit a sequence
of bids after provisional allocations and prices for the previous phase are
revealed; such designs often result in long and time-consuming procedures which
are justified for one-time spectrum auctions but unsuitable for Internet
advertisement\footnote{One can eliminate the interactive aspect of package
bidding auction by using a {\em proxy agent}, as Ausubel and Milgrom discuss in
Section 3.4 of \cite{AM02}. While this technique eliminates the communication
burden, it is not enough to achieve incentive-compatibility.}. Moreover, because of the fast-paced nature of online
advertisement, one cannot expect bidders to reach an equilibrium outcome for
each individual ad auction if the underlying auction is not a truthful one.

In this paper we investigate whether it is possible to design direct-revelation
incentive-compatible auctions whose revenue is competitive against a core outcome
(we call such auctions {\em core-competitive} auctions). In core-competitive
auction design, we seek to relax (a) instead of (b) above. More precisely, we
define {\em core revenue benchmark} as the smallest revenue among all the
core-outcomes. We say that an auction is $\alpha$-core-competitive if its
revenue is at least an $1/\alpha$ fraction of the core revenue benchmark.

We formally define the notion of core-competitiveness in section 2, and later we
focus on the design of core-competitive auctions for the Text-and-Image setting.
\comment{ which is at the same time the simplest and most-common in practice
settings in which complementarities are present in Internet Advertisement.} We
give a randomized universally-truthful mechanism which is $O(\ln \ln
k)$-core-competitive, where $k$ is the number of slots. We also give a lower
bound showing that this factor is tight. We note that in ad auction settings,
there are several repeated auctions with each auction generating only a small
revenue. For such settings, a seller care about the overall performance and
therefore randomized auctions are perfectly fine from a practical auction
design perspective. We also study deterministic auctions since for some settings
randomization may not be desired; for instance, for one time auctions like
spectrum auctions. We give a deterministic mechanism which is
$O(\sqk)-$core-competitive, and again show that this factor is tight
for deterministic mechanisms.

Finally, to the best of our knowledge, the notion of core-competitiveness has
not been studied before. It is our belief that developing tools and techniques
for designing core-competitive auctions, and understanding the possibilities and
limitations of such auctions, will be very useful from a practical auction
design perspective.

\subsection{Related Work}

The line of inquiry that seeks to design package auctions that implement
core outcomes in equilibrium was started by Ausubel and Milgrom \cite{AM02}.
This line has been further developed in
\cite{DM08, ausubel2010core, day2012quadratic, erdil2010new,
goeree2009impossibility,lamy2010core}. The authors design an iterative procedure
that asks bidders in each round for packages they want to bid on as well as bid
values for each of those packages. In each round a set of provisionally winning bids are
identified. This proceeds until no further bids are issued in a given round. Our
work differs from this line of work in the sense that we require
incentive-compatibility; in the core-selecting package auctions literature, the
focus is on implementing core outcomes {\em in equilibrium}.

Another stream of related work is the design of incentive compatible auctions
that tries to optimize for revenue in a prior-free setting. This research direction was initiated in
\cite{GoldbergHW01,FiatGHK02,GoldbergH03} and resulted in a sequence of followup results which are too large to survey here. We refer to Hartline's book 
\cite{hartline2013mechanism} for a comprehensive discussion. The first successful results gave
auctions for the digital goods that approximate the $\mathcal{F}^2$ revenue benchmark, the maximum revenue one can extract from at least two players using
fixed prices. More modern versions of this result
\cite{HartlineY11,HaH13,DevanurHH13} compare against the envy-free benchmark
(how much revenue it is possible to extract from an outcome where any two agents
wouldn't like to swap places). This resulted in success stories for a large
class of environments such as multi-units, matroids and permutation
environments. \comment{One salient feature of this line of work is that all environments under
consideration are {\em symmetric}, i.e., if a certain allocation is feasible,
then the same allocation with identities of the players permuted are also
feasible. Asymmetric environments, such as single minded combinatorial
auctions, can only be handled by ``symmetrizing'' the environment. This is done, for
example, by randomly permuting the identities of the agents.} Our work differs from the above line of work as we consider environments with complementarities, while the envy-free revenue literature mostly focused on environments with substitutes. In Section
\ref{subsec:comparison} we discuss in detail the relation between the envy-free
benchmark and the core-revenue benchmark and we argue that the core-revenue
benchmark captures some of the {\em no-envy} notions.

Closer to our line of inquiry is the work of
\cite{MicaliValiant} and \cite{aggarwal2006knapsack}. In \cite{MicaliValiant}, they design revenue extraction mechanisms for general combinatorial auctions where their benchmark is the maximum social welfare extractable from all except one player (the one with the top bid). They use randomization to obtain a mechanism with $O(\log n)$ approximation factor. They also give a matching lower bound of  $\Omega(\log n)$ for randomized mechanisms, and for deterministic mechanisms they give a lower bound of $\Omega(n)$. \cite{aggarwal2006knapsack} study knapsack auction where there are $k$ identical items and each bidder demands a certain number of them. Their benchmark is a version of envy-free pricing where a bidder has to pay at least as much as the bidders with lower demands \footnote{This is also called monotone benchmark, see also \cite{LR12} and \cite{BEJ+13} for its definition on digital goods auction.}. They get an approximation ratio of $\alpha \cdot \OPT - \lambda O(\log \log \log n)$ where $\OPT$ is the optimal envy-free revenue, $\alpha$ is a constant number and $\lambda$ equals to the highest valuation of any bidder. Although their approach is useful when $\lambda$ is much smaller than $\OPT$; it performs poorly when $\lambda$ is close to $\OPT$ which can be the case in the Image-and-Text auction.

We note that the revenue benchmarks of both the above papers are stronger than the core-benchmark. Thus, one might wonder if the mechanisms proposed in  \cite{MicaliValiant} and  \cite{aggarwal2006knapsack} perform better against the core benchmark? However, one can show that mechanisms given in both the above papers perform worse than our mechanism when compared to the core benchmark. The mechanism of \cite{aggarwal2006knapsack} can perform arbitrarily bad compared to the core benchmark, and the mechanism of \cite{MicaliValiant} still gets only $O(\log n)$ using randomization when compared to the core benchmark\footnote{We refer to Section \ref{subsec:comparison} for further discussion on this.}. In some sense, this suggests that a too strong benchmark that leads to large lower bounds in approximation ratio impedes the design of a good revenue-maximizing mechanism. We believe that the core benchmark is a more fundamental benchmark (as argued in series of papers starting with the work of \cite{ausubel2006lovely}), and as our work show, it looks amenable to a good multiplicative approximation ratio.

Finally, while we focus on the {\em forward} setting (i.e. an auctioneer selling
goods to various buyers), there is a very extensive literature on the
procurement (reverse auction) version of this problem (i.e. a buyer purchasing
goods from various sellers). In this line of work, the goal is to design
procurement auctions where the total amount paid by the buyer approximates a
certain {\em frugality benchmark}. This line of work was initiated in
\cite{archer2007frugal} in which the frugality benchmark is defined as the best solution
after the agents in the optimal solution are removed. A more sophisticated frugality benchmark was
introduced in \cite{karlin2005beyond}. Their benchmark can be seen as the
counterpart of the core-revenue benchmark in procurement settings. Frugality in
the procurement setting is also a topic which is too broad to be completely
covered here, but we mention a few recent papers on the topic:
\cite{kempe2010frugal,chen2009frugal,elkind2007frugality,immorlica2010first}.

\section{Preliminaries}

\subsection{Core Outcomes}\label{subsec:core_1}

We consider set $N = \{1, \hdots, n\}$ of single-parameter agents
with value $v_i$ for being allocated and value zero otherwise. The set of
feasible allocations is specified by an {\em environment}, which is a
collection of subsets of players that can be
simultaneously allocated $F \subseteq 2^N$. We say that an environment is
{\em downward-closed} if every subset of a feasible set is also feasible, i.e.,
$X \in F$ and $Y \subseteq X$ imply $Y \in F$.

An outcome in such environment is a pair $(X, p)$ where
$X \in F$ corresponds to the selected set of players and
$p \in \R^N$ is a vector of (possibly negative) payments.
Players have {\em quasi-linear} utility functions,
i.e., $u_i(X, p) = v_i - p_i$ if $i \in X$ and $u_i(X, p) = -p_i$
otherwise. We also define the utility of the auctioneer as its
revenue $u_0(X, p) = \sum_{i=1}^n p_i$.

Throughout this paper, given a vector $v \in \R^N$ and $S \subseteq N$,
we define $v(S) := \sum_{i \in S} v_i$.\\

We can associate with the single parameter setting described above a {\em
coalition value function} $w:2^{\bar{N}} \rightarrow \R_+$ (where $\bar{N} =
\{0\} \cup N$) given by:
\[
w(S) = \begin{cases}
\max_{X \in F, X \subseteq S, p \in \R^N_+} \sum_{i \in  S} u_i(X, p) & 0 \in S\\
0 & 0 \not\in S
\end{cases}
\]
for every $S \subseteq \bar{N}$. The pair $(\bar{N}, w)$ defines a {\em cooperative
game} with transferable utility. The coalition value of a set corresponds to
the total utility that can be obtained by a certain set by defecting from the
rest of the agents. Clearly, a coalition that doesn't contain the auctioneer
can't obtain any value. A coalition containing the auctioneer can obtain utility
equal to $\max_{X \in F, X \subseteq S, p \in \R^N_+} \sum_{i \in S} u_i(X, p) = \max_{X \in F, X \subseteq S}  v(X)$.

An imputation of utilities for a coalition $S \subseteq \bar{N}$ corresponds to
a vector of utilities $(u_i)_{i \in S}$ specifying how the coalition value is
split between the agents, in other words, a vector $u_i \geq 0, \forall i \in S$
and $\sum_{i \in S} u_i \leq w(S)$. We say that an imputation of utilities for
$\bar{N}$ is in the {\em core} if no coalition can defect and produce an
imputation of utilities that is better for all agents in the coalition.
Formally:

\begin{definition}[core] Given a cooperative game $(\bar{N}, w)$ we define the core as
  the following set of utility imputations:
  $$\core(F,v) =
  \left\{ u \in \R^{\bar{N}}_+; \sum_{i=0}^n u_i = w(\bar{N}) \text{
  and } w(S) \leq \sum_{i \in S} u_i, \forall S \subseteq \bar{N} \right\} $$
\end{definition}

Notice that $w(S) \leq \sum_{i \in S} u_i$ is a necessary and sufficient
condition for $S$ not wanting to defect. We say now that an outcome $(X, p)$
is in the core if the utilities produced are in $\core(F,v)$. Precisely:

\begin{definition}[core outcomes]\label{defn:core_1}
  Given a single parameter setting $F$ and
  valuation profile $v$, an outcome $(X, p)$ is in the core if the vector of
  utilities is in $\core(F,v)$.
\end{definition}

The following are important properties of core outcomes:
\begin{enumerate}
\item A core outcome
is also a social welfare maximizing outcome, since
$ \sum_{i \in X} v_i = \sum_{i=0}^n u_i = w(\bar{N}) = \max_{X^* \in F}
v(X^*)$;
\item The core is always non-empty, since the following allocation is always in
  the core: $(X^*, p)$ where   $X^*$ maximizes $v(X)$ and $p_i = v_i$ 
  or $i \in X^*$ and $p_i = 0$ otherwise;
\item Given a utility imputation $u \in \core(F,v)$, there is a core outcome that realizes this vector:
select a set $X^* \in F$ maximizing $\sum_{i \in X^*} v_i$ and allocate
to $X$ and charge prices $p_i = v_i - u_i$ for $i \in X^*$ and $p_i = 0$
otherwise. The outcome clearly realizes utilities for $i \in X^*$. For $i \notin
X^*$, notice that  $w(\bar{N}) = \sum_{i=0}^n u_i =
(u_0 + \sum_{i \in X^*} u_i + \sum_{i \in N \setminus X^*} u_i) \geq
w(\bar{N}) + \sum_{i \in N \setminus X^*} u_i$. So for all $i \in N
\setminus X^*$, $u_i = 0$;
\item If the environment $F$ is downward-closed, then for every $u \in
  \core(F,v)$ there is an outcome with  non-negative payments that realizes it.
  The construction is the same as in the previous item. Note that if $F$ is
  downward closed, $X^* \setminus i \in F$ for every $i \in X^*$, therefore:
  $v(X^*) = u_0 + u(X^*) \geq u_i + v(X^* \setminus i)$ so $u_i \leq v_i$
  and hence $p_i = v_i - u_i \geq 0$.
\end{enumerate}

The previous observations allow us to rephrase Definition \ref{defn:core_1} in a
more direct way. Notice that in the following definition,
$v(S \setminus X) \leq p(X \setminus S)$ is a simple rephrasing of
the $w(S) \leq \sum_{i \in S} u_i$ condition.

\begin{definition}[core outcomes - rephrased]\label{defn:core_2}
  Given a single parameter setting $F$ and
  valuation profile $v$, an outcome $(X, p)$ is in the core if $p_i \leq v_i$ for all $i
  \in N$ and \comment{  $X \in \text{argmax}_{X \in F} v(X)$ and} for all $S \in F$, $$v(S \setminus X) \leq p(X \setminus S)$$
\end{definition}

Definition \ref{defn:core_2} allows for a natural interpretation of the core in
auction settings. If an outcome is not in the core, then there is a set $S$ with
$v(S \setminus X) > p(X \setminus S)$, which means that agents in $S \setminus
X$ could come to the auctioneer and offer him to evict agents $X \setminus S$
and allocate to them instead, since they are able to collectively pay the
auctioneer more than the revenue he is getting from $X \setminus S$. This
characterizes core outcomes as outcomes for which no negotiation is possible
between the auctioneer and losing coalitions

\subsection{Core-revenue benchmark}

The discussion after Definition \ref{defn:core_2} shows that whenever an outcome
is not in the core, the auctioneer can potentially raise his revenue by
negotiating with losing coalitions. This suggests that the revenue of the core
might be a natural benchmark against which to compare. We define as follows:

\begin{definition}
  Given a single parameter setting $F$ and a valuation profile $v$, we define
  the core revenue benchmark as:
  $$\cb(F,v) := \min \{ u_0| u \in \core(F,v) \}.$$
\end{definition}

Consider for example the case of multi-unit auctions, which can be modeled by $F
= \{ X \subseteq N; \abs{X} \leq k \}$ for some fixed constant $k< n$ and agents
sorted such that $v_1 > v_2 > \hdots > v_n$. It is straightforward from
Definition \ref{defn:core_2} that an outcome is in the core iff it allocates to
$X = \{1,\hdots, k\}$ and if $p_i \geq v_{k+1}$ for $i \in X$. Notice that
the revenue from core outcomes range from $k \cdot v_{k+1}$ all the way to
$\sum_{i=1}^k v_k$. The core benchmark corresponds to the minimum revenue of a
core outcome, so for multi-unit auctions $\cb(F,v) = k \cdot v_{k+1}$.

It is not a coincidence that this is the same revenue as the VCG auction.
In fact, it is a well-known fact that the core revenue is always at least
the VCG revenue. This holds with equality when $F$ is a matroid. For an
in-depth discussion on the relation between the VCG mechanism and the core we
refer the reader to Ausubel and Milgrom \cite{AM02} and Day and Milgrom
\cite{DM08}.

\begin{lemma}[\cite{AM02}]
  For any environment $F$ and any valuation profile $v$, the price paid by any
  agents in a core outcome is at least his VCG price. This implies in particular
  that  the core revenue   benchmark is at least the revenue of the VCG mechanism:
  $$\cb(F,v) \geq \vcgr(F,v) := \sum_{i \in X^*} [v(X^*_{-i}) - v(X^*) + v_i]$$
  where $X^* = \text{argmax}_{X \in F} v(X)$ and $X^*_{-i} = \text{argmax}_{X
  \in F, i \notin F} v(X)$. Moreover, if $F$ is a matroid, the the above
  expression holds with equality.
\end{lemma}

\begin{proof}
  If $(X, p)$ is a core outcome, by the condition in Definition
  \ref{defn:core_2}, $v(X^*_{-i} \setminus X^*) \leq p(X^* \setminus
  X^*_{-i})$, which can be re-written as: $v(X^*_{-i}) - v(X^*) \leq - [ v(X^*
  \setminus X^*_{-i}) - p(X^* \setminus X^*_{-i}) ] \leq v_i - p_i$. So
  $p_i \geq v(X^*_{-i}) - v(X^*) + v_i$ which is the revenue that the VCG
  mechanism extracts from player $i$.

  If $F$ is a matroid, then for each $i \in X^*$, $X^*_{-i}$ is of the form
  $X^*_{-i} = X^* \cup j \setminus i$ and therefore the VCG payments are given
  by   $p_i = \max\{ v_j; j \notin   X^*; X^* \cup j \setminus i \in F \}$.
  Now, we show that the VCG outcome is in the core: for any matroid
  basis $S \in F$,   there is a   one-to-one mapping between
  $\sigma: S\setminus X \rightarrow X \setminus S$ such that for
  $i \in S$ with $v_i > 0$, $X \cup i \setminus \sigma(i) \in F$, therefore,
  $p_{\sigma(i)} \leq v_{i}$. Summing this inequality for all $i \in S$ we
  obtain the core condition in Definition \ref{defn:core_2}.
\end{proof}

The previous lemma says that when there is {\em substitutability} among agents,
the core revenue benchmark is exactly the VCG revenue. When there are
complementarities, however, the core revenue benchmark can be arbitrarly higher
than the VCG revenue. Consider for example the famous example of 
\cite{AM02, ausubel2006lovely} in which there are $3$ players and $2$ items: the
first player has a valuation of $1$ for the first item, the second player has
a valuation of 1 for the second item and the third player has a valuation of $1$
for getting both items. This example can be translated to our setting by taking
the environment to be $F = \{\emptyset, \{1\}, \{2\}, \{3\}, \{1,2\} \}$. The VCG auction
allocates $X = \{1,2\}$ and charges zero payments. So, $\vcgr(F,v) = 0$. The
core revenue, however, is equal to one ($\cb(F,v) = 1$) since by
taking the condition in Definition \ref{defn:core_2} with $X = \{1,2\}$ and $S =
\{3\}$, we get: $p_1 + p_2 \geq v_3 = 1$.

\subsection{Core competitive auctions}

Our goal in this paper is to be able to truthfully extract revenue that is
competitive with the core-revenue benchmark. An auction for the single parameter
setting consists of two mappings: (i) {\em allocation function}, that maps a
profile of valuation functions to a distribution over allocations $x: \R^n_+
\rightarrow \Delta(F)$, where $\Delta(F)$ denotes the set of probability
distributions over $F$; (ii) {\em payment function}, that maps a profile of
valuation functions to the expected payment of each agent: $p: \R^n_+
\rightarrow \R^N_+$.

We abuse notation and define the maps $x_i : \R^N_+ \rightarrow [0,1]$
as the probability of winning for player $i$, i.e., $x_i(v) = \Prob[i \in
X(v)]$. A mechanism is said to to be {\em individually rational} if for all
profiles $v$, $u_i(v) = v_i x_i(v) - p_i(v) \geq 0$. A mechanism is said to be
{\em incentive-compatible} (a.k.a. {\em truthful}) if agents maximize their
utility by reporting their true value. In other words:
$$v_i x_i(v) - p_i(v) \geq v_i x_i(v'_i, v_{-i}) - p_i(v'_i, v_{-i})\quad \forall v'_i$$

The following lemma due to Myerson \cite{M81} gives necessary and sufficient
conditions for an auction to be individually rational and incentive compatible:

\begin{lemma}[\cite{M81}]
\label{thm:tf}
  A mechanism defined by maps $x$ and $p$ is individually rational and incentive
  compatible if: (i) for every $i$ and fixed valuations $v_{-i}$ for other
  players, $v_i \mapsto x_i(v_i, v_{-i})$ is monotone non-decreasing; (ii) the
  payment function is such that $p_i(v_i) = v_i x_i(v_i, v_{-i}) -
  \int_0^{v_i} x_i(u, v_{-i}) du$.
\end{lemma}

Our goal in this paper is to study auctions whose revenue is competitive with
the core-revenue benchmark.

\begin{definition}[core competitive auctions]
  We say that an auction defined by $x, p$ is $\alpha$-core competitive if for
  every profile of valuation functions $v \in \R^N_+$, $$\sum_i p_i(v) \geq
  \alpha^{-1} \cdot \cb(F,v).$$
\end{definition}

\subsection{Comparison with other benchmarks}\label{subsec:comparison}

A natural question at this point is how does the core benchmark compare with
other revenue benchmarks. Perhaps one of the closest benchmarks in this spirit is the
envy-free benchmark, which corresponds to the minimum revenue of an allocation
for which an agent would not want to trade positions with a different agent. This
benchmark has been successfully used in various papers (
\cite{GuruswamiHKKKM05, HartlineY11, HaH13, DevanurHH13} to cite a few) to design
approximately-optimal revenue-extracting mechanisms. This benchmark, however, is
very appropriate for {\em symmetric} settings, i.e., a setting in which whenever
an allocation is feasible, a similar allocation with the names of agents
permuted is also feasible. For asymmetric settings, however, it is not clear
what the envy-freedom condition means since some agents can't be simply replaced
by others. In an ad auction where ads can be either texts (occupying one slot)
or images (occupying multiple slots), it is not clear how to define what the
envy of an image for a text means, since the image is not able to replace a
single text.

On the other hand, however, the core-revenue benchmark captures some notion of
``envy", which is made explicit in Definition \ref{defn:core_2}. One can think
of the inequality in the defintion as the ``envy'' of an allocated image
for a group of allocated text ads. Or more generaly, as the ``envy'' of a set of
losing players for a set of winning players that they can replace.
What the core benchmark doesn't capture, however, is the ``envy" from one
allocated agent for another allocated agents.
For this reason, for symmetric settings, the envy free benchmark can be
arbitrarly higher than the core-revenue benchmark, which boils down to the VCG
revenue, as discussed in Section \ref{subsec:core_1}.

Another important benchmark against which to compare is the one introduced by
Micali and Valiant \cite{MicaliValiant}. Given any feasiblity set, the
authors define as the maximum social welfare obtainable after the largest valued
agent is excluded. Formally:  $$\mv(F,v) = \max_{X \in F, i^* \notin X} v(X)$$
where $i^*$ is the agent with largest value\footnote{The benchmark of
\cite{MicaliValiant} is defined for a generic multi-parameter setting. For the
exposition, we specialize it for the single-parameter setting we are studying.}.

\begin{lemma}
  For any environment $F$ and any valuation profile $v$, the core revenue
  benchmark is dominated by the Micali-Valiant benchmark: $$\mv(F,v) \geq
  \cb(F,v).$$
\end{lemma}

\begin{proof}
  Let $(X, p)$ be the outcome of the VCG auction. Now, define $p'$ such that
  $p'_i = v_i$ if $i \in X \setminus i^*$, $p'_{i^*} = p_{i^*}$ and
  $p'_i = 0$ otherwise. First we show that $(X, p')$ is in the core. Notice
  that if $i^* \notin X$, then $p'(X \setminus S) = v(X \setminus S) \geq v(S
  \setminus X)$ so clearly $(X,p')$ is in the core. If $i^* \in X$, then
  $p'_{i^*} = v(X_{-i^*}) - v(X \setminus i^*)$ where $X_{-i^*}$ is the
  allocation with $X \in F, i^* \notin X$ maximizing $v(\cdot)$. Therefore
  $p(X) = v(X\setminus i^*) + p'_{i^*} =  v(X_{-i^*})$ therefore, $p(X
  \setminus S) = v(X_{-i^*}) - p(X \cap S) \geq v(S) - v(X \setminus S) = v(S
  \setminus X)$.   Finally, notice that $\cb(F,v) \leq p'(X) = \mv(F,v)$.
\end{proof}

Micali and Valiant \cite{MicaliValiant} give an individually rational and
incentive compatible randomized mechanism whose revenue is an $O(\log n)$
approximation of $\mv(F,v)$ and that such approximation factor is tight.
This directly translates in a same factor approximation for the
core-revenue benchmark. They also show that no deterministic auction can
approximate $\mv(F,v)$ by a factor better then $\Omega(n)$.

One reason for which it is hard to improve the $\mv$-benchmark even for very
simple settings, $\mv$ is too stringent: for example, for the digital goods
setting $F = 2^N$, $\mv(F,v) = \sum_i v_i - \max_i v_i$. Indeed, both lower
bounds in \cite{MicaliValiant} are given for the digital goods setting.
For this setting, the core-revenue benchmark is zero, since there is no
natural competition among the agents.

We believe that the core revenue benchmark provides a more achievable goal and
therefore a more likely avenue for improvement for particular settings.
For the Text-and-Image setting, for example, the lower bounds of
\cite{MicaliValiant} imply that no mechanism can approximate the $\mv$-benchmark
by a better factor then $\Omega(\log k)$ for randomized mechanisms and
$\Omega(k)$ for deterministic mechanisms. For the $\cb$-benchmark, however, we
are able to obtain $O(\ln \ln k)$ and $O(\sqk)$ respectively.

The core revenue also has the
important property of disentangling the problems of achieving high revenue
for setting with substitutes and for settings with complements, since the former
becomes trivial under the $\cb$-benchmark while the latter is quite
challenging. Under the $\mv$-benchmark, both substitutes and complements are
challenging.

\section{$O(\sqk)$-core-competitive auction for Text-and-Image setting}

\subsection{Text-and-Image Setting}

Consider $k$ advertisement slots and $n$ bidders. Each bidder either corresponds
to a text ad, which demands one slot to be displayed, or an image, which demands
all $k$ slots. It is public information that whether each bidder is a text or an
image. Each bidder's value for being displayed is given by $v^T_i$ for text ads
and $v^I_i$ for image ads. The values are private information of the bidders.

Let $n^T$ and $n^I$ be the number of text and image ads respectively. We assume
w.l.o.g. that $n^T \geq k+1$ and $n^I \geq 2$ (adding a few extra bidders with
value zero if necessary). We also assume that the indices of the players are
sorted such that valuations of text ads are
$v_1^T \geq v_2^T \geq \hdots \geq v_{n^T}^T$
and valuations of image ads are
$v_1^I \geq v_2^I \geq \hdots \geq v_{n^I}^I$. For convenience, we define the
{\em maximum extractable revenue of text ads} as:
$$\Phi^T := \max_{j\in\{1..k\}} j \cdot v^T_j$$
We will also denote the $k$-th harmonic partial sum by $\H_k = \sum_{j=1}^k
\frac{1}{j} = O(\ln k)$. It is a well known fact that 
\begin{equation}
\label{eq:mert}
\Phi^T \geq \frac{1}{\H_k} \sum_{j=1}^k v_j^T,
\end{equation}
 since $j \cdot v^T_j \leq \Phi^T$ for all
$j$, so $\frac{1}{j} \Phi^T \geq v_j^T$. We finish the argument by summing
the previous inequality for all $j = 1..k$.

\subsection{A deterministic core-competitive auction}

We start by presenting a $O(\sqk)$-core competitive deterministic
auction. We will use this mechanism as a building block for the more complicated
randomized mechanism given in Section \ref{sec:rand-ita}. As a first step,
we provide a characterization of the core-revenue in that setting:

\begin{lemma} 
\label{lem:ita:corerevenue}
Given a Text-and-Image setting, if the highest value feasible set consists of text ads ($\sum_{i=1}^k v^T_i \geq v_1^I$) then $\cb(F,v) = \max\{
k v^T_{k+1}, v_1^I\}$. If the highest value feasible set consists of an image ad,
then $\cb(F,v) = \max\{v^I_2, \sum_{i=1}^k v^T_i\}$.
\end{lemma}

\begin{proof}
Note that in Text-and-Image setting the winner set cannot contain both text and image ads.
Now consider the special case where $\sum_{i=1}^k v^T_i = v_1^I$. In this case no matter from which group is the winning set, the sum of payments has to be at least $\sum_{i=1}^k v^T_i = v_1^I$. Because if the sum of payments is less, then the non-winning group can offer more to the auctioneer and all of them benefit more.

Now consider the case where $\sum_{i=1}^k v^T_i > v_1^I$. In this case the winners are the first $k$ text ads with sum of valuations $\sum_{i=1}^k v^T_i$. In order to be a core outcome, the sum of payments of the winners has to be more than valuations of image ads and hence more than $v^I_1$. The payment of each winner also has to be more than the valuation of the highest text ad who is not in the winning set which is $v^T_{k+1}$. Therefore, the sum of payments of the winners has to be more than $k \cdot v^T_{k+1}$. We conclude for this case that $\cb(F,v) = \max\{k v^T_{k+1}, v_1^T\}$.

Now consider the case where $\sum_{i=1}^k v^T_i < v_1^I$. In this case, the winner is an image ad with value $v^I_1$. In order to be a core outcome, the payment of the winner has to be at least the value of the second best image ad which is $v^I_2$. The payment of the winner also has to be more than the sum of valuations  of the highest $k$ text ads which is $\sum_{i=1}^k v^T_i$. We conclude that $\cb(F,v) = \max\{v^I_2, \sum_{i=1}^k v^T_i\}$.
\end{proof}

Recall the example by Ausubel and Milgrom discussed in the introduction:
if we have two text ads and one image ad all with value $1$, the text ads are
selected and their payment is zero. The reason for that is that if any text ad
decreases his value all the way to $\epsilon > 0$, the text ads are still
selected. One way to get around this problem is picking the allocated set in
such a way that a decrease in value for any given text significantly decreases
the likelihood of the entire set being picked.

A natural way to do so is to allocate to the set which has the potential of
generating the largest revenue. One proxy for that is the maximum extractable
revenue $\Phi^T$ which corresponds to the maximum revenue you can extract
by setting a uniform price. This motivates the mechanism that allocates to the
highest value image ad if $v_1^I \geq \Phi^T$ and otherwise allocates to
the $j$ highest text ads where $j$ is the maximum index such that $j v_j^T \leq
v_1^I$. Here the payments are according to critical prices.

In the Ausubel and Milgrom example, for instance,
the text ads are still allocated but their threshold is now $\frac{1}{2}$, so
their total revenue is $1$. This mechanism is clearly truthful since the
allocation is monotone and its revenue is clearly an improvement over VCG. The
gap between its revenue and the core-revenue benchmark can be as bad as $O(\ln
k)$. Consider the following example: one image ad with value $\H_k$ and
$k$ text ads with value $1/i$ for $i=1, \hdots, k$. The core-benchmark is $\H_k$
but the revenue of the mechanism is only $\Phi^T = 1$.

A way to improve this mechanism is to increase the weight attributed to the text
ads by a factor of $\sqk$. Now, we are ready to define our mechanism:

\begin{framed}
{\bf Allocation rule:} If $v^I_1 \geq \Phi^T \cdot \sqk$ then allocate to the
highest value image ad. Otherwise, allocate to the $j$ text ads with
largest values where $j$ is the largest $j \leq k$ such that $j \cdot v^T_j \geq v_1^I /
\sqk$.

{\bf Pricing rule:} Allocated bidders are charged according to critical values.
\end{framed}

\begin{lemma} 
\label{lem:dcv}
In the deterministic Text-and-Image mechanism, if the first image ad
wins, her critical value is $\max \{ v_2^I, \Phi^T \cdot \sqk \}$. If
a set of $j$ text ads win, their critical value is $\max\{ v^T_{k+1}, v_1^I / (j
\cdot \sqk)\}$.
\end{lemma}

\begin{proof}
Recall that the critical value of each winner is the minimum bid for which she remains a winner fixing the other bidders' bids. 

The case where the winner is an image ad is easy to proof. Note that in this case the winner has value $v^I_1$. The minimum bid in order to remain the winner has to be at least the value of the second highest image ad which is $v^I_2$ and has to be larger than $\Phi^T \cdot \sqk$ to win against text ads. Therefore, the critical value of the winner is $\max\{v^I_2, \Phi^T \cdot \sqk\}$.

Now we consider the case where the winners are the $j$ highest text ads. If $\max\{v^T_{k+1}, v^I_1/(j \cdot \sqk)\}$ is equal to $v^T_{k+1}$, we have $v^I_1/\sqk \leq k \cdot v^T_{k+1}$ hence $j = k$ by the way we select $j$. Hence, the first $k$ text ads win. Moreover, the winners' payments has to be at least $v^T_{k+1}$ in order to be in the first $k$ text ads, therefore, the critical value of the winners is $\max\{v^T_{k+1}, v^I_1/(j \cdot \sqk)\} = v^T_{k+1}$.

If $\max\{v^T_{k+1}, v^I_1/(j \cdot \sqk)\}$ is equal to $v^I_1/(j \cdot \sqk)$, we prove by contradiction that the critical value of winners is $v^I_1/(j \cdot \sqk)$. Lets assume that there exist value $v'$ ($v' < v^I_1/(j \cdot \sqk)$) such that if a winner (W) bids $v'$, she remains in the winning set and hence $v'$ is her critical value. Let $j'$ be the number of winners when W bids $v'$. We know that the value of $\Phi^T$ is at most $j' \cdot v'$ since W is in the winning set. The value of $\Phi^T$ has to be greater than $v^I_1 / \sqk$ in order for text ads to win against image ads. Therefore, we have $j' \cdot v' \geq v^I_1 / \sqk$. On the other hand we have $v' < v^I_1/(j \cdot \sqk)$ which implies $j  \cdot v'< v^I_1/ \sqk$. Hence we conclude that $j' > j$ which contradicts with the fact that $j$ is the largest number such that $j \leq k$ and  $j \cdot v_{j}$ is larger than or equal to $\frac{v^I_1}{\sqk}$.
\end{proof}

Using the previous two lemmas we prove the following theorem and finish this section.

\begin{theorem}
The deterministic Text-and-Image mechanism is $O(\sqk)$-core competitive.
\end{theorem}

\begin{proof}
We prove the theorem by considering two cases: Case (i) when the first image ad wins and Case (ii) when the first $j$ text ads win.

In Case (i) the winner is the image ad with value $v^I_1$ and his payment $\max\{v^I_2, \Phi^T \cdot \sqk\}$ by Lemma \ref{lem:dcv} is the revenue of our deterministic Text-and-Image mechanism. The value of $\cb$ in this case is $\max\{v^I_2, \sum_{i=1}^k v^T_i\}$ by Lemma \ref{lem:ita:corerevenue}. Therefore, using Equation \ref{eq:mert} we conclude that the revenue of our deterministic Text-and-Image mechanism is at least $\sqk$ fraction of $\cb$.

In Case (ii) the winners are the first $j$ text ads. By Lemma \ref{lem:dcv} we know that their critical value is $\max\{ v^T_{k+1}, v_1^I / (j\cdot \sqk)\}$. If their critical value is equal to $v_1^I / (j\cdot \sqk)$ then the total revenue of the mechanism is $v_1^I / \sqk$. If their critical value is equal to $v^T_{k+1}$ then it means that $v^T_{k+1}\geq v_1^I / (j\cdot \sqk)$, hence $j$ is equal to $k$ since  $j$ is the largest $j \leq k$ such that $j \cdot v^T_j \geq v_1^I /\sqk$. Therefore, the total revenue is $k \cdot v^T_{k+1}$. As a result the total revenue in Case (ii)  is $\max\{v_1^I / \sqk, k \cdot v^T_{k+1}\}$. The value of $\cb$ in this case is $\max\{k v^T_{k+1}, v_1^I\}$ by Lemma \ref{lem:ita:corerevenue}. Therefore the revenue of our deterministic Text-and-Image mechanism is at least $\sqk$ fraction of $\cb$.
\end{proof}

\subsection{A $O(\sqk)$ lower bound for deterministic mechanisms}

Now we show that $O(\sqk)$ is necessary for deterministic core-competitive
mechanisms. Formally, we show that no mechanism that is anonymous and satisfies
independence of irrelevant alternatives can provide an approximation ratio
better then $O(\sqk)$. A word of caution: while anonymity and independence of
irrelevant alternatives are commonly used assumptions in lower bounds
for deterministic mechanisms \cite{AshlagiDL12},
they are not completely innocuous as shown by \cite{AggarwalFGHIS11}.

\begin{definition} A mechanism ($\M = (x, p)$) is anonymous if the following
holds. Let $v$ and $v'$ be two valuation profiles that are permutations of each
other (i.e. the set of valuations are the same but the identities of bidders are
permuted). Say $v = \operatorname{permutation}(v')$. If  $x(v) = S_1$ and $x(v')
= S'$, then $S' = \operatorname{permutation}(S)$.  \end{definition}

\begin{definition} \label{def:anonymity} A mechanism ($\M = (x, p)$) satisfies
  independence of irrelevant alternatives if we decrease the bid of a losing
  participant, it does not hurt any winner. More formally, for every valuation
  profile $v$ and loser participant $i \not\in x(v)$, if we decrease the value
  of $i$ from $v^T_i$ to $\hat{v_i}^T < v_i$ then $x(v) \subseteq x(\hat{v_i}^T,
  v_{-i})$.  \end{definition}

\begin{theorem}
Let $M^*$ be a deterministic mechanism with optimum core competitive factor
satisfying anonymity and independence of irrelevant alternatives. Then
there exist a valuation profile for which revenue of $M^*$ is at most $\sqk$ of
$\cb$.
\end{theorem}

\begin{proof}
Let valuation profile $v$ consists of $k$ text ads  $\{v^T_1, v^T_2, \ldots
v^T_k\}$ where value $v^T_i$ is equal to $1/i$ and $2$ image ads $\{v^I_1,
v^I_2\}$ both with value $\sqk$. Now we consider two cases: 

{\bf Case (i) $M^*$ allocates to an image ad.} 
Note that the revenue of $\M^*$ is the payment of the winner and
is at most $\sqk$. Now, lets  increase the valuation of the winner to $\ln(k)$
and build a new valuation profile $v'$. Note that by Lemma \ref{thm:tf} the
winner and his payment in $v'$ remains the same as in $v$. Therefore, the
revenue of $v'$ is $\sqk$ while its $\cb$ by Lemma \ref{lem:ita:corerevenue} is
$\ln(k)$.\\

{\bf Case (ii) $M^*$ allocates to a set of text ads .}
We build a group of $k$ valuation profiles $v^{(1)}, \ldots,
v^{(k)}$ and show that in at least one of them the difference between $\cb$ and
revenue of $\M^*$ is $\sqk$. valuation profile $v^{(1)}$ is the same as $v$ and
we build $v^{(i+1)}$ from $v^{(i)}$ by the following procedure. If text ad
$v^T_{i+1}$ is a winner in $v^{(i)}$ then we obtain $v^{(i+1)}$ by increasing
value of $v^T_{i+1}$ to one in $v^{(i)}$. Otherwise,  if text ad $v^T_{i+1}$ is
a loser in $v^{(i)}$ then we obtain $v^{(i+1)}$ by decreasing value of
$v^T_{i+1}$ to zero in $v^{(i)}$.

Let $j$ be the largest number such that $j\leq k$ and text ad $v^T_j$ is a
winner in $v^{(j)}$. Now we claim that every text ad $j'$ where $j' > j$ is a
loser in $v^{(j)}$. Otherwise, if such $j'$ exist then $j'$ will also be a
winner in $v^{(j')}$ since by independence of irrelevant alternative $j'$
remains a winner in all valuation profiles $v^{(\ell)}$ for $j < \ell < j'$.
This contradicts with the fact that $j$ is the largest number. Therefore, we
know that in valuation profile $v^{(j)}$ all the winners are between $1$ and $j$
and hence we have at most $j$ winners. Note that $v^{(j)}$ is obtained from
$v^{(j-1)}$ by increasing the value of $v^T_j$ from $1/j$ to $1$ and by Lemma
\ref{thm:tf} his payment is at most $1/j$. Also, all the winners in $v^{(j)}$
have valuation $1$, so we claim that all the winners should pay
the same amount. Before proving the claim, we show that this is enough
to finishes the proof in this case. Mechanism $\M^*$ at valuation profile
$v^{(j)}$ has at most $j$ winners each paying at most $1/j$, therefore,
the revenue of $\M^*$ is $1$ while $\cb$ of $v^{(j)}$ is
$\sqk$ (by Lemma \ref{lem:ita:corerevenue}).\\

We finish this section by proving the claim that the payments of winners of $\M^*$
at valuation profile $v^{(j)}$ are all the same.  Assume otherwise and let $a$ and
  $b$ be two text ads in the valuation profile $v$ where both are winners but
  they pay different amounts. w.l.o.g. assume $p_a < p_b$. Lets pick value $x$
  such that $p_a < x < p_b$ and $x$ be different than all the valuations in
  $v^{(j)}$. Note that such $x$ exists since there are finite number of bidders
  in $v^{(j)}$ but infinitely many numbers in range $(p_a, p_b)$. Now if we
  decrease the valuation of bidder $v^T_a$ from $1$ to $x$ and obtain valuation
  profile $A$ she remains a winner by Lemma \ref{thm:tf}. If we decrease the
  valuation of bidder $v^T_b$ from $1$ to $x$ and obtain valuation profile $B$
  she does not remain a winner by Lemma \ref{thm:tf}. Note that the single
  bidder in $A$ with valuation $x$ is a winner but the single bidder with
  valuation $x$ in $B$ is not a winner while $A$ and $B$ are permutations of
  each other. This contradicts with anonymity (see Definition \ref{def:anonymity}) of
  $\M^*$.  \end{proof}

\section{A randomized $O(\lln)$-core competitive
mechanism}\label{sec:rand-ita}

In this section we improve the $O(\sqk)$-core competitive mechanism presented in
the last section with the use of randomization. Recall that in the deterministic
mechanism we decide on allocating to text or image ads based on the ratio $v_1^I
/ \Phi^T$ being above or below $\sqk$. If we allow randomness, we can decide a
threshold as a random function of this ratio.  Optimizing the revenue as a
function of this distribution, we obtain the following mechanism:

\begin{framed} {\bf Allocation rule:} Consider the ration $\psi = v^I_1 /
\Phi^T$: \begin{itemize} \item[$\star$] if $\psi \leq 2$ allocate the items to
      the $j$ largest text ads, where $j$ is the largest number such that $j
      v^T_j \geq v^I_1/2$.  \item[$\star$] if $2 < \psi$, allocate to the
    highest valued image ad with probability $\min\{1, \ln(\psi) / \ln(\ln
k)\}$. With the remaining probability, leave the items unallocated.\\
\end{itemize}

{\bf Pricing rule:} Allocated bidders are charged according to Myerson's
integral.  \end{framed}

In the following lemma we calculate the critical values of winners and total
revenue of our randomized mechanism.  \begin{lemma} \label{lem:revr} The revenue
  of our mechanism is the following.  \begin{equation*} \sum_i p_i(v) =
    \begin{cases} \max\{k \cdot v^T_{k+1}, v_1^I/2\} & \text{case (i): } \psi <
    2 \\ (v^I_1 + 2 \Phi^T \ln(2) - 2\Phi^T) / \lln & \text{case (ii): } 2 \leq
    \psi \leq \ln(k)\\ (\ln(k) \cdot \Phi^T + 2 \Phi^T \ln(2) - 2\Phi^T) / \lln
    &  \text{case (iii): } \psi > \ln(k) \end{cases} \end{equation*}
	
\end{lemma} \begin{proof} We consider the following three cases.

{\bf Case (i).} In this case we have $j$ text winners.
  We prove that the critical value of each of them is at least $v_1^I/(2j)$.
  Suppose not and assume that the critical value of text ad $A$ is $v_A'$ where
  $v_A' < v_1^I/(2j)$. This means that when $A$ bids $v_A'$ she still remains a
  winner. Therefore, there exists a number $j'$ such that  
  \begin{equation}
 j' \cdot v_A' > v_1^I/2   \label{eq:jrm1}
  \end{equation} 
  in order for text ads to win against image ads. Using $v_A' < v_1^I/(2j)$ and Equation \ref{eq:jrm1}
  we conclude that $j' > j$ which contradicts with the fact that $j$ is the
  largest number that $j v^T_j > v_1^I/2$. Therefore the critical value of each
  of the $j$ text winners is at least $v_1^I/(2j)$. Moreover if $k v^T_{k} >
  v_1^I/2$ then each winner's critical value must be more than $v^T_{k+1}$ in
  order to be in the winning set. Therefore, the critical value of the winners
  is equals to $\max\{v^T_{k+1}, v_1^I/(2j)\}$ and the total revenue in this
  case is $\max\{k \cdot v^T_{k+1}, v_1^I/2\}$.\\
 
{\bf Case (ii).} In this case the image ad with largest valuation $v^I_1$ wins
and his expected payment is the expected total revenue of our mechanism.
\begin{align*} p(v^I_1) &= v^I_1 x^I_1(v^I_1, v_{-1}) - \int_{2\Phi^T}^{v^I_1}
x^I_1(u, v_{-1}) du & \text{Lemma \ref{thm:tf}}\\ &= v^I_1
\ln(v^I_1/\Phi^T)/\lln - \int_{2\Phi^T}^{v^I_1} \ln(u/\Phi^T)/\ln \ln (k) du &
\text{replacing} x^I_1\\ &= v^I_1 \ln(v^I_1/\Phi^T)/\lln - \left[ (u \cdot
\ln(u/\Phi^T) - u)/\lln \right]_{2\Phi^T}^{v^I_1}   & \text{solving the
integral}\\ &= (v^I_1 + 2 \Phi^T \ln(2) - 2\Phi^T) / \lln \end{align*}\\

{\bf Case (iii). } Note that if $v^I_1$ is larger than $\ln(k) \cdot \Phi^T$
then her probability of winning is one. Therefore, his payment will be the same
as when her valuation is $\ln(k) \cdot \Phi^T$. Therefore, using case (ii) the
payment $v^I_1$ in this case is $(\ln(k) \cdot \Phi^T + 2 \Phi^T \ln(2) -
2\Phi^T) / \lln$.  \end{proof}


\begin{theorem} Core competitive factor of randomized Image-and-Text mechanism
is $\max\{2, 1.43\cdot \lln\}$.  \end{theorem} \begin{proof} We prove the theorem by
considering three cases similar to Lemma \ref{lem:revr}.

{\bf Case (i) : $\psi < 2$.} By Lemma \ref{lem:ita:corerevenue} we know that if $\sum_1^k
v^T_i \geq v^I_1$ then $\cb$ is equal to $\max\{k v^T_{k+1}, v_1^I\}$. As the
revenue of our mechanism in this case is $\max\{k \cdot v^T_{k+1}, v_1^I/2\}$
(by Lemma \ref{lem:revr}) the proof of the lemma follows. If $\sum_1^k v^T_i <
v^I_1$ then by Lemma \ref{lem:ita:corerevenue} we know that $\cb = \max\{v^I_2,
\sum_{i=1}^k v^T_i\}$ which is at most $v^I_1$. Therefore, the core competitive
factor for this case is $2$ and the proof of the lemma follows.\\

{\bf Case (ii) : $2 \leq \psi \leq \ln(k)$.}
By Lemma \ref{lem:ita:corerevenue} we know that if $\sum_1^k
v^T_i \geq v^I_1$ then $\cb$ is equal to $\max\{k v^T_{k+1}, v_1^I\}$. As
$\Phi^T \geq k v^T_{k+1}$ and $v^I_1 \geq 2\Phi^T$, we conclude that $\cb$ is at
most $v^I_1$. If $\sum_1^k v^T_i < v^I_1$ then by Lemma
\ref{lem:ita:corerevenue} we know that $\cb = \max\{v^I_2, \sum_{i=1}^k v^T_i\}$
which is at most $v^I_1$. Therefore, in this case $\cb$ is at most $v^I_1$.  The
revenue of our mechanism in this case is $(v^I_1 + 2 \Phi^T \ln(2) -
2\Phi^T)/\lln \simeq (v^I_1 - 0.61\Phi^T)/\lln$ (by Lemma \ref{lem:revr}). As
$v^I_1 \geq 2\Phi^T$, the revenue of our mechanism is at least $(v^I_1 -
0.61v^I_1/2)/\lln = 0.695 v^I_1 / \lln$, hence it is at least $0.695/ \lln$
fraction of $\cb$ (\ie{}     $1.43 \cdot \lln$-core competitive) and the proof
of the lemma follows.\\

{\bf Case (iii): $\psi > \ln(k)$.} 
In this case we have $v^I_1 > \ln(k) \cdot \Phi^T$ which by
Equation \ref{eq:mert} implies $v^I_1 \geq \sum_1^k v^T_i$. Hence the $\cb$ in
this case is $\max\{v^I_2, \sum_{i=1}^k v^T_i\}$ which is at most $v^I_1$. The
rest of the proof is similar to case (ii) and the core competitive factor for
this case is at least $1.43\cdot \lln$.  \end{proof}

\section{A Lower Bound for Revenue of Randomized Mechanisms in Image-and-Text setting}
In this section we prove lower bound of $\Omega(\lln)$ for core-competitive factor of
randomized mechanisms. The structure of the proof is as follows. Let assume
$\R^* = (x^*, p^*)$ to be a truthful randomized mechanism (satisfying conditions
of Lemma \ref{thm:tf}) with optimum core-competitive factor. \comment{
Note that here we prove for the mechanisms that are truthful in expectation
which are more general than truthful in universal sense.}
We derive a distribution over valuation profiles for the Text-and-Image
setting such that the expected revenue of $\R^*$
is at most $2$ and the expected value of $\cb$ is $\Omega(\lln)$. Therefore, we
conclude that for at least one of the valuation profiles in the support of
$\alpha$, $R^*$ yields a revenue that is smaller than core revenue by factor
$\Omega(\lln)$.\\

{\bf A distribution over valuation profiles}. Given $k$ text ads and one image
ad, define a distribution $\D$ over valuation profiles as the following.
The value of each text ad is taking iid from the set
$\{1, \frac{1}{2}, \ldots, \frac{1}{k}\}$, each element has probability
$\frac{1}{k}$. The value of the image ad is taken from set $\{H,
\frac{H}{2}, \ldots, \frac{H}{H} \}$ were each element has probability
$\frac{1}{H}$, where $H = \cH$.\\

In the following lemma we prove that the expected revenue of $\R^*$ is at most
$2$.

\begin{lemma} \label{lem:exprev} The expected revenue of $\R^*$ for $\alpha$ is
at most $2$.  \end{lemma}


\begin{proof}
  From the perspective of any given player, a randomized mechanism can be seen
  as a random threshold being offered to $i$ as a function of $v_{-i}$. So the
  revenue that can be extracted from each agent $i$ in expectation, is the
  revenue that can   be extracted from $i$ by using a random threshold, which
  is the maximum   revenue that can be obtained from any given player by a
  fixed threshold (since   the revenue from a random threshold is the
  expectation   of revenue that can be obtained from a fixed
  threshold)\footnote{Here is a simple mathematical derivation of those
    arguments for differentiable allocation function $x(v)$ (since monotone functions are almost-everywhere differentiable,
    the same argument can be  easily extended just by performing the equivalent
    calculations on  discontinuities) given an allocation $x(x)$, let
    $\hat{x}(v_i) = \E_{v_{-i}}x(v_i, v_{-i})$, then the expected revenue that
    can be extracted from agent $i$ with distribution $F$ is given by
    $p_i = \E_{v_i} [\int_0^{v_i} u \cdot \partial
\hat{x}(u) du = \int_0^{\infty} \int_0^{v_i} u \cdot \partial \hat{x}(u) du dF(v)
 $. Inverting the order of the integration we get:  $p_i = \int_0^{\infty}
 \int_u^\infty  u \cdot \partial \hat{x}(u) dF(v) du = \int_0^{\infty}  u \cdot
 \partial \hat{x}(u) (1-F(u)) du  \leq \max_u [u \cdot (1-F(u)) ] \cdot
 \int_0^\infty \partial \hat{x}(u) du \leq \max_u [u \cdot (1-F(u)) ]$, which
 is the maximum revenue obtained from a single threshold.}.

  It is simple to see that under $\D$ the best revenue that can be obtained by a
  single threshold from any given text ad is $1/k$ and the revenue that can be
  obtained from an image is $1$. So, the total revenue is at most $k \cdot
  \frac{1}{k} + 1 = 2$.
\end{proof}

\begin{lemma} \label{lem:expcore} The expected value of the core revenue
  benchmark is doubly-logarithmic:  $\E_{v \sim \D }\cb(v)
\geq \Omega( \lln )$. \end{lemma}

\begin{proof} Throughout this proof, let $v$ be a random variable drawn from
  $\D$.  For any given text ad, $\E[v_i^T] = \ncH / k$. Now, we
  bounds its variance by: $$\Var[v_i^T] = \E[(v_i^T)^2] - \E[v_i^T]^2 \leq
    \E[(v_i^T)^2]  = \frac{1}{k} \sum_{j=1}^k \frac{1}{j^2} \leq \frac{\pi^2}{6\cdot
    k} \leq \frac{2}{k}.$$ Therefore, $\E[\sum_i v_i^T] = \cH$ and $\Var[\sum_i
  v_i^T] \leq 2$. By Chebyshev's inequality $$\Prob \left( \left\vert \sum_i
    v_i^T - \ncH
  \right\vert \geq 2 \right) \leq \frac{1}{2}.$$

  By Lemma \ref{lem:ita:corerevenue} we know that the $\cb(v) = \min \{ \sum_i
v_i^T, v^I \}$. Now, we are ready to lower bound the core revenue benchmark:
$$\begin{aligned}
\E[\cb(v)] & = \E\left[ \min\left(\textstyle\sum_i v_i^T, v^I\right) \right] \\
 &\geq \frac{1}{2} \cdot \E\left[\min\left(H - 2, v^I \right)\right] &\text{by Chebyshev's inequality} \\
 &\geq \frac{1}{2} \cdot \frac{1}{H}  \sum_{i=1}^H \min\left(H - 2, \frac{H}{i}  \right) &\text{replacing $v^I$}\\
 &= \Omega(\log H)
\end{aligned}$$
Since $H = O(\log k)$ we get that $\E[\cb(v)] \geq \Omega(\lln)$.
\end{proof}
\comment{

In the following inequalities we lower bound the expected
value of $\cb$.  \begin{align} \notag \E[\min(X, \alpha_I)] &= \sum_{i =
1}^{\cH} \frac{1}{\cH} \min(X, \cH/i)\\ \notag &\geq \frac{1}{\cH} \sum_{i =
5}^{\cH}  \min(X, \cH/i) & \text{forgetting about first 4}\\ \notag &\geq
\frac{1}{\cH} \sum_{i = 5}^{\cH} \frac{1}{4} \cH/i & \text{by Equation
\ref{ineq:cheb}}\\ \label{ineq:tightcr} &\geq \sum_{i = 5}^{\cH} \frac{1}{4} 1/i
\end{align} Note that the above inequality implies that $\E[\min(X, \alpha_I)]$
is roughly $\ln\ln(k) - \ln(\ln(5))$ which is at least $\frac{1}{5} \lln$ for
large enough $k$.}

\begin{theorem} The core-competitive factor of $\R^*$ is at
least $\Omega( \lln)$.  \end{theorem} 

\begin{proof} 
  Since $\E[\cb(v)] = \Omega(\lln)$ and $\E[\sum_i p_i(v)] =
  O(1)$, it follows from the probabilistic methods that there must be at least
  one valuation profile for which $ \cb(v) \geq \Omega(\lln) \cdot \E[\sum_i
  p_i(v)]$.
\end{proof}

\paragraph{Note on inefficient allocations:} The auctions described in this
paper implement outcomes that are often not socially optimal. Moreover, even
when more then one socially optimal allocation is available, the mechanism might
allocate to an agent that is part of no efficient allocation. This is unlike,
for example, the Micali-Valiant mechanism \cite{MicaliValiant} which always
allocates to a (random) subset of the agents allocated by the VCG mechanism.
Next we show that sometimes allocating to agents which are not allocated in any
efficient outcome is necessary in order to get core-competitiveness better
then $O(\ln k)$.

\begin{theorem}
  Any mechanism for the Text-and-Image setting that only allocates for a subset
  of the agents selected by the VCG mechanism has $\Omega(\ln k)$ core competitive hardness.
\end{theorem}

\begin{proof}
  Consider $k$ text ads with $v_i^T$ drawn from the same distribution used for
  the previous lower bound and one image ad with $v_1^I = \H_k / 2$.
  Using the expectation and variance of $\sum_i v_i^T$ computed
  earlier in this section, we know by Chebyshev's inequality that $\Pr( \abs
  \sum_i v_i^T - \H_k \vert > \H_k / s ) \leq \Omega(1/\H_k^2)$. So the image
  ad is allocated with probability $O(1/\H_k^2)$.  Since the
  revenue obtained from any given text ad in expectation is at most $1/k$ (by
  Lemma \ref{lem:exprev}), the total revenue is at most $k \cdot \frac{1}{k} +
  O(\frac{1}{\H_k^2}) \cdot \H_k / 2 = O(1)$. The expected core revenue
  benchmark, however, is at least $(1-\frac{1}{\H_k^2}) \cdot \H_k / 2 =
  \Omega(\ln k)$.
\end{proof}

\bibliographystyle{alpha}
\bibliography{refs}

\end{document}